\documentclass[11pt]{article}             % Use with LaTeX2e
\usepackage{osid}                         %

\usepackage{amsmath,amssymb,amscd,amsthm,mathrsfs}
\usepackage{bbm}
\usepackage{tikz-cd}

\usepackage[bookmarks=false,pdfstartview={FitH}]{hyperref}

\theoremstyle{plain}% Theorem-like structures provided by amsthm.sty
\newtheorem{lemma}{Lemma}
\newtheorem{thm}{Theorem}

\theoremstyle{definition}

\theoremstyle{remark}

\newcounter{app}
\renewcommand*{\theapp}{\Alph{app}}

\newcommand*\app[1]{%
 ~\refstepcounter{app}\label{app_#1}\hypertarget{app:#1}{\theapp}}
\newcommand*\appref[1]{\hyperlink{app:#1}{\ref*{app_#1}}}

\title{Quantum-Dynamical Semigroups and the Church of the\\Larger Hilbert Space} %: The Compact Case}
\author{Frederik vom Ende%\thanks{Supported by ...} 
	\\[1mm]{\footnotesize\it Technische Universit{\"a}t M{\"u}nchen, 
	School of Natural Sciences, 85747 Garching and\\
   Munich Centre for Quantum Science and Technology (MCQST) \& Munich Quantum Valley (MQV), Schellingstra{\ss}e 4, 80799 M{\"u}nchen, Germany\\ {frederik.vomende@gmail.com}}\\[2ex]
}

\AtBeginDocument{
  \label{CorrectFirstPageLabel}
  
}

\begin{document}

\maketitle
\begin{abstract}
In this work we investigate Stinespring dilations of quantum-dynamical semigroups,
which are known to exist by means of a constructive proof given by Davies in the early 70s.
We show that if the semigroup describes an open system, that is, if it does not consist of only unitary channels, then the evolution of the dilated closed system has to be generated by an unbounded Hamiltonian;
subsequently the environment has to correspond to an infinite-dimensional Hilbert space, regardless of the original system.
Moreover, we prove that the second derivative of Stinespring dilations with a bounded total Hamiltonian yields the dissipative part of some quantum-dynamical semigroup -- and vice versa.
In particular this characterizes the generators of quantum-dynamical semigroups via Stinespring dilations.
%Finally we work through a simple example to understand Davies' original construction in order to understand why unbounded Hamiltonians are necessary for dilations which describe a quantum-dynamical semigroup.
\end{abstract}

\section{Introduction}
Completely positive maps, which are among the fundamental objects in modern quantum physics and quantum information theory, admit the following central representations: the Choi \cite{Choi75} or Choi-Jamio\l{}kowski \cite{Jamiolkowski72} matrix, the operator-sum form using Kraus operators \cite{Kraus71}, and the Stinespring representation \cite{Stinespring55}.
Each of these comes with their own advantages; to name just a few: complete positivity is checked most easily via the Choi matrix, the task of numerically generating random quantum maps is done most efficiently via
%generating 
random Kraus operators \cite{Kukulski21}, and the Stinespring dilation has a direct physical interpretation which can also be applied to certain experimental setups \cite{Braun01,Haake10}.
Details regarding these representations can be found in Appendix \protect\appref{A}, and for the reader's convenience we sketched the relation between these concepts
in Figure~\ref{fig_choi_kraus_stinespring}.
For more details on these representations from the point of view of quantum information theory we refer to \cite[Ch.~8.2]{NC10} \& \cite[Ch.~4.2 ff.]{Heinosaari12}.
\begin{figure}[!ht]
\begin{center}
\begin{tikzcd}
\text{Choi-Jamio\l{}kowski}\quad \arrow[r, "\text{diagonalize}", shift left] & \quad\text{Kraus}\ \qquad \arrow[l, "\text{vectorize}", shift left] \arrow[r, "\substack{\text{collect Kraus op.}\\\text{in larger matrix}}", shift left] & \ \qquad\text{Stinespring} \arrow[l, "\substack{\text{first ``block-column''}\\\text{of larger matrix}}", shift left]
\end{tikzcd}
\end{center}
\caption{
Interconversion scheme between the Choi-Jamio\l{}kowski matrix, the Kraus operators, and the Stinespring representation. 
%Details can be found in Appendix \protect\appref{A}.
}\label{fig_choi_kraus_stinespring}
\end{figure}
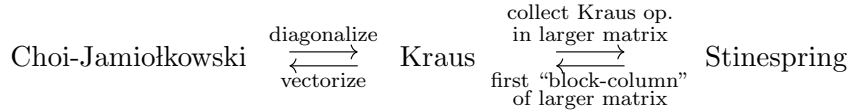

Note that when we say Stinespring dilation we do not mean the general theorem for completely positive maps 
%\cite{Stinespring55} 
but rather the special form it gets if the map in question $\Phi$ is not only completely positive but also trace-preserving. Loosely speaking there then exists an environment such that the action of $\Phi$ can be written as the restriction of some unitary action on the larger (closed) system, cf.~\eqref{eq:Stinespring_cptp} in Appendix~\appref{A}.
While this result does not involve the notion of time, fifty years ago it was Davies who first proved that an analogous result holds if $\Phi$ is replaced by a dynamical process $(\Phi_t)_{t\geq 0}$, assuming it acts on a finite-dimensional system and is ``memory-less'' \cite{Davies72}.
Note that every such dynamical process is a semigroup of contractions
% (as a consequence of the Russo-Dye theorem \cite{Russo66}) 
so there always exists a one-parameter unitary group which is a dilation of the original semigroup \cite[Ch.~I, Theorem 8.1 \& Ch.~III.9]{SzNagyFoias}.
However the strength of Davies' result is that his dilation admits more structure, and that his proof is constructive.
%\footnote{
%Davies works in Heisenberg picture, but the formulation we gave is the equivalent statement in the Schr\"odinger picture, see also \cite[Coro.~2 \& Remark 5]{vE_dirr_semigroups}.
%}
%
%we will work in Schr\"odinger picture (Heisenberg picture equivalent)
%
%w Then Stinespring's dilation theorem ............
%
%describe how these three representations can be converted into each other? 
Recently Burgarth et al.~have improved Davies' result by explicitly constructing a dilation of the same type which can even be dynamically decoupled \cite{Burgarth22},
complementing a result of Gough and Nurdin that there also exist dilations which cannot be dynamically decoupled \cite{Gough17}.

Although quantum processes which have a ``memory'' are of interest in modern quantum theory \cite{Wolf08b,Plenio_NMarkov_Review2014,Bhatta20,Ptas22,Spaventa22}
the fundamental importance of memory-less processes as well as their simple, yet rich mathematical structure encourages us to further investigate Davies' result and, more generally, investigate if, when, and how such processes can be cast into the framework of Stinespring dilations.
\section{Main Results}\label{sec:main}
While our main results revolve around quantum dynamics in finite dimensions, we will need general Hilbert spaces as well as (operators on) operators on such spaces as auxiliary objects.
Thus to set the framework -- orienting ourselves towards \cite[Ch.~15 \& 16]{MeiseVogt97en} -- given a real or complex Hilbert space $\mathcal K$ we denote the bounded linear operators on $\mathcal K$ by $\mathcal B(\mathcal K)$, and $\mathcal B^1(\mathcal K)$ is the collection of all trace-class operators on $\mathcal K$, i.e.~all compact operators\footnote{
A linear map $A $ between normed spaces is called compact if the closure of the image of the closed unit ball under $A $ is compact.
Now if $A $ operates between Hilbert spaces over the same field, then compactness is equivalent to a representation $A =\sum_{n\in N}s_n|f_n\rangle\langle g_n|$, $N\subseteq\mathbb N$ where $\{f_n\}_{n\in N}$, $\{g_n\}_{n\in N}$ are orthonormal systems in the respective Hilbert space and $\{s_n\}_{n\in N}\subset(0,\infty)$ are the unique singular values of $A $ \cite[Prop.~16.3]{MeiseVogt97en}.
}
on $\mathcal K$ the 
%sequence of 
singular values of which sum up to a finite value.
Topological structure on those spaces is induced by their ``defining'' norms: $\mathcal B(\mathcal K)$ is a Banach space together with the usual operator norm $\|\cdot\|_\infty$ and $\mathcal B^1(\mathcal K)$ becomes a Banach space when equipped with the trace norm $\|\cdot\|_1$ which is defined as the sum of the singular values of the input operator.
The name ``trace class'' is justified by the fact that it is meaningful to define the trace of such operators in the usual way, that is, as $\sum_{i\in I}\langle f_i,(\cdot)f_i\rangle$ with $(f_i)_{i\in I}$ any orthonormal basis of the underlying Hilbert space.
The trace then has the same properties as in finite dimensions, it is a continuous linear form on $\mathcal B^1(\mathcal K)$, and, most importantly, the trace class is a two-sided ideal in the bounded operators.
A subset of the trace class most important for quantum theory are the states, which are defined as the positive semi-definite trace-class operators of trace one.
Finally, going one level higher we will also be concerned with linear maps between the bounded operators or between trace classes.
As these are linear maps between normed spaces the natural choice to quantify boundedness is via the operator norm induced by the norms on domain and co-domain.
Similar to \cite{PG06} we will denote the operator norm of linear maps between trace class operators (bounded operators) by $\|\cdot\|_{1\to 1}$ ($\|\cdot\|_{\infty\to\infty}$).

Having established notation let us now come to this work's two main objects: given a map $\Phi$ from $[0,\infty)$ into the linear maps on $\mathbb C^{n\times n}$
\begin{itemize}
\item one calls $\Phi$ a \textit{quantum-dynamical semigroup}
(for short: \textsc{qds}) 
if it is a one-parameter semigroup ($\Phi(t)\circ\Phi(s)=\Phi(t+s)$ for all $s,t\geq 0$) of completely positive (cf.~Appendix~\appref{A}) trace-preserving maps which is strongly continuous at zero\footnote{
This means $\lim_{t\to 0^+}\|\Phi(t)(A)-A\|_1=0$ for all $A\in\mathbb C^{n\times n}$.
However, while this is the standard formulation of continuity of quantum-dynamical semigroups, because we are concerned with operators $\Phi_t$ which have finite-dimensional domain this is the same as usual norm continuity $\lim_{t\to 0^+}\|\Phi(t)-\operatorname{id}_{n}\|_{1\to 1}=0$ \cite[Prop.~2.1.20 (iv,b)]{vE_PhD_2020}.
}.
For convenience we will henceforth write $\Phi_t:=\Phi(t)$, as well as $\textsc{cptp}(n)$ for the collection of all completely positive trace-preserving linear maps on $\mathbb C^{n\times n}$.
Now the generators of such semigroups have been classified by the celebrated result of Gorini, Kossakowski, Sudarshan, and Lindblad \cite{GKS76,Lindblad76}. They showed that $\Phi$ is a \textsc{qds} if and only if there exists $H\in\mathbb C^{n\times n}$ Hermitian as well as a finite collection of matrices $\{V_j\}_{j\in J}\subset\mathbb C^{n\times n}$ such that $\Phi_t=e^{tL}$ for all $t\geq 0$ where
$
L=-i\operatorname{ad}_{H}-{\bf\Gamma}
$
and
\begin{equation}\label{eq:lindblad_V}
{\bf\Gamma}=\sum_{j\in J}\Big( \frac12 \big(V_j^* V_j (\cdot)+(\cdot) V_j^* V_j\big)-V_j(\cdot) V_j^* \Big)\,.
\end{equation}
%is of \textsc{gksl}-form.
Sometimes we refer to \textsc{qds}s as (time-independent) Markovian, cf.~\cite{Wolf08a}.
\item we call $\Phi$ a \textit{Stinespring curve}
if there exists a complex Hilbert space $\mathcal K$, a self-adjoint operator $H$ on $\mathcal K$, and a state $\omega$ on $\mathcal K$
%(i.e.~$\omega$ is a trace-class operator on $\mathcal K$ with trace one such that $\omega$ is positive semi-definite) 
such that
\begin{equation}\label{eq:Stinespring_curve}
\Phi_t\equiv\operatorname{tr}_{\mathcal K}\big(e^{iHt}\big((\cdot)\otimes\omega\big)e^{-iHt}\big)
\end{equation}
for all $t\geq 0$.
Here $\operatorname{tr}_{\mathcal K}:\mathcal B^1(\mathbb C^n\otimes\mathcal K)\to\mathbb C^{n\times n}$ is the usual partial trace over $\mathcal K$, that is, $\operatorname{tr}_{\mathcal K}(A)$ is the unique $n\times n$ matrix which satisfies
$
\operatorname{tr}(B\operatorname{tr}_{\mathcal K}(A))=\operatorname{tr}((B\otimes\mathbbm1_{\mathcal K})A)
$ for all $B\in\mathbb C^{n\times n}$.
Moreover we call a Stinespring curve \textit{type~I} if
%, additionally, 
there exists a bounded self-adjoint operator $H$ such that \eqref{eq:Stinespring_curve} holds,
%$t\mapsto e^{iHt}$ is norm-continuous is zero (equivalently, if $H$ is bounded \cite[Thm.~13.36]{Rudin91}), 
and \textit{type~II} otherwise.
%Indeed Davies' proof shows that one can even choose $\omega$ to be a pure state.
\end{itemize} 
The aim of this article is to clarify the relation between these two concepts.
After all, the Stinespring dilation theorem for \textsc{cptp} maps (cf.~Eq.~\eqref{eq:Stinespring_cptp} in Appendix~\appref{A}) shows that every quantum map can be interpreted as the restriction of a larger closed system where the latter consists of the original system together with a sufficiently large environment. 
Thus is it most reasonable to assume that this interpretation continues to hold if the (stationary) quantum map is replaced by a suitable dynamical process.

Our first result (Theorem~\ref{thm_main} below) shows that, while finite-dimensional quantum-dynamical semigroups are known to be Stinespring curves, they are of type~I (if and) only if they describe a closed system.
In other words casting a dynamical (Markovian) open system interaction into the Stinespring framework forces the dynamics of the larger closed system to be generated by an unbounded Hamiltonian.
In particular the environment used for such a Stinespring dilation has to be infinite-dimensional.
We remark it has been argued that purely exponential decay (i.e.~Markovian dynamics) can only occur if the Hamiltonian of system plus environment is unbounded
below and above \cite{Beau17}. However this has been shown to only be partially true as a dilation for the example of qubit phase damping has been constructed where the overall (unbounded) Hamiltonian is positive, hence bounded from below \cite{Burgarth17}.

Thus while not entirely new, we nonetheless state the following as a theorem; on the one hand it serves to clarify the relation between our two main objects, and on the other hand our rather simple proof will motivate our second main result.
\begin{thm}\label{thm_main}
Every finite-dimensional quantum-dynamical semigroup is a Stinespring curve. Moreover, it is of type~I if and only if it is unitary at all times.
%In particular any ancilla Hilbert space $\mathcal K$ then has to be infinite-dimensional.
\end{thm}
\begin{proof}[Proof idea]
The existence result was first due to Davies \cite{Davies72}, see also \cite[Ch.~9.4]{Davies76}.
%For conveying the idea let us assume that $\omega$ is a pure state, that is, $\omega=|\psi\rangle\langle\psi|$ for some $\psi\in\mathcal K$ with $\|\psi\|=1$; 
The only non-trivial statement left to show is that 
if a \textsc{qds} $(\Phi_t)_{t\geq 0}$ is at the same time a type~I Stinespring curve (i.e.~\eqref{eq:Stinespring_curve} holds for some bounded operator $H$), then $\Phi_t$ has to be unitary for all $t\geq 0$.
Our strategy is to
% via a set of Kraus operators of $\Phi_t$ which are analytic in $t$: 
%First define $g_1:=\psi$ and complete the orthonormal system $\{g_1\}$ to  an orthonormal basis $\{g_j\}_{j\in J}$ of $\mathcal K$ (i.e.~$1\in J$).
%Then one can show that $\{V_j^* e^{iHt}V_1\}_{j\in J}$ set of Kraus operators of $\Phi_t$ for all $t\geq 0$ where $V_j:\mathbb C^n\otimes\mathbb C^n\otimes\mathcal K$ is the isometry $x\mapsto x\otimes g_j$ for all $j\in J$.
%In particular every Kraus operator is differentiable so by the product rule
%\begin{align*}
%\frac{d}{dt}\Phi_t\Big|_{t=0}&=\sum_{j\in J}\Big(\frac{d}{dt}V_j^* e^{iHt}V_1\Big|_{t=0}\Big)(\cdot)(V_j^* e^{-iH\cdot 0}V_1)^*\\
%&\qquad+\sum_{j\in J}(V_j^* e^{iH\cdot 0}V_1)(\cdot)\Big(\frac{d}{dt}V_j^* e^{-iHt}V_1\Big|_{t=0}\Big)^*\\
%&=\sum_{j\in J}(V_j^* (iH)V_1)(\cdot)(V_j^*V_1)^*+i\sum_{j\in J}(V_j^*V_1)(\cdot)(V_j^* (-iH)V_1)^*\,.
%\end{align*}
%Because $\{g_j\}_{j\in J}$ is an orthonormal basis of $\mathcal K$ one finds $V_j^*V_1=\delta_{j1}\operatorname{id}$ so
%\begin{align*}
%\frac{d}{dt}\Phi_t\Big|_{t=0}=i\big(V_1^* HV_1(\cdot)-(\cdot)V_1^* HV_1\big)=i\big[V_1^* HV_1,\,\cdot\,\big]
%\end{align*}
%is given by (${i}$ times) the commutator with the Hermitian matrix $V_1^*HV_1$.
differentiate \eqref{eq:Stinespring_curve}:
\begin{align*}
\frac{d}{dt}\Phi_t\Big|_{t=0}&=\operatorname{tr}_{\mathcal K}\Big( \frac{d}{dt}e^{iHt}\Big|_{t=0}((\cdot)\otimes\omega)e^{-iH\cdot 0} \Big)\\
&\qquad+\operatorname{tr}_{\mathcal K}\Big(e^{iH\cdot 0} ((\cdot)\otimes\omega) \frac{d}{dt}e^{-iHt}\Big|_{t=0}\Big)\\
&=\operatorname{tr}_{\mathcal K}\big(iH((\cdot)\otimes\omega)\big)+\operatorname{tr}_{\mathcal K}\big(((\cdot)\otimes\omega)(-iH)\big)\,.
\end{align*}
This works because $\operatorname{tr}_{\mathcal K}$ is a continuous linear operator and because $H$ is bounded.
At this point note that the linear map $X\mapsto \operatorname{tr}_{\mathcal K}\big(iH((X\otimes\omega)\big)$ (resp.~$X\mapsto \operatorname{tr}_{\mathcal K}\big((X\otimes\omega)(-iH)\big)$) on $\mathbb C^{n\times n}$ is nothing but the left (resp.~right) multiplication with the matrix $i\operatorname{tr}_\omega(H)$ (resp.~$-i\operatorname{tr}_\omega(H)$), cf.~Lemma~\ref{lemma_rewrite_trace_expression} in Appendix~\appref{B}.
Here $\operatorname{tr}_\omega(H)$ is the Hermitian $n\times n$
matrix which satisfies
$$
\operatorname{tr}\big( \operatorname{tr}_\omega(H)X \big)=\operatorname{tr}\big(H(X\otimes\omega)  \big)
$$
for all $X\in\mathbb C^{n\times n}$ (called ``partial trace of $H$ with respect to $\omega$''), cf.~also \cite[Ch.~9, Lemma 1.1]{Davies76}.
Therefore $\frac{d}{dt}\Phi_t|_{t=0}\equiv i[\operatorname{tr}_\omega(H),\,\cdot\,]$. But by assumption $(\Phi_t)_{t\geq 0}$ is a \textsc{qds} so its derivative at zero is its \textsc{gksl}-generator $L$ \cite{GKS76,Lindblad76}, hence
$\Phi_t\equiv e^{it\operatorname{ad}_{\operatorname{tr}_\omega(H)}}= e^{it \operatorname{tr}_\omega(H)}(\cdot)e^{-it \operatorname{tr}_\omega(H)}$
showing that $\Phi_t$ is a unitary channel for all $t\geq 0$.
%because $\operatorname{tr}_\omega(H)$ is Hermitian.
The detailed proof can be found in Appendix~\appref{B}.
\end{proof}
Put simply, the reason why quantum-dynamical semigroups which describe an open system have to be type~II Stinespring curves is that the first derivative of type~I curves corresponds to closed system dynamics. More precisely the first derivative of \eqref{eq:Stinespring_curve} -- assuming type~I -- is given by $i\operatorname{ad}_{\operatorname{tr}_\omega(H)}$
which on the other hand has to be the generator of the semigroup itself.
We emphasize that Theorem 1 continues to hold for Markovian dynamics on any time-interval $[0,t_f]$, $t_f>0$ because our argument relies on a quantity that is local (at time zero).
%which, unsurprisingly, most of the information is contained in the ``blocks'' of the Hamiltonian of the composite system.

So far we did not see any dissipative effects of Stinespring curves. However, such curves---by design---model the interaction with an environment so there have to be irreversible parts to its action.
This has to do with the well-known fact that the closed system part is first order in $t$, while actual environment interaction appears only from the second derivative onward.
This is what the following result will be about: not only does the second derivative of type~I Stinespring curves look like the generator of a (purely dissipative) \textsc{qds}, but every such semigroup is the second derivative of some type~I Stinespring curve: 

%the ``reason'' behind this result is that in Stinespring, Hermitian part is $\mathcal O(t)$ whereas dissipative part is only $\mathcal O(t^2)$. While the former is shown in the proof of Theorem~\ref{thm_main}, the latter comes from the following correspondence between generators of quantum-dynamical semigroups and second derivatives of Stinespring curves:

\begin{thm}\label{thm_main_2}
Given any $n\in\mathbb N$ the following statements hold.
\begin{itemize}
\item[(i)] For every type~I Stinespring curve $\Phi$
%$(\Phi_t)_{t\geq 0}$ 
%into $\textsc{cptp}(n)$ 
there exists a set $\{V_j\}_{j\in J}\subset\mathbb C^{n\times n}$ with $|J|\leq\min\{\operatorname{rk}(\omega)\operatorname{dim}(\mathcal K),n^2\}$
%---but $J$ is at most countable---
such that
\begin{equation}\label{eq:stinespring_second_der}
\ddot\Phi_0=-\sum_{j\in J}{\bf\Gamma}_{V_j}\,.
\end{equation}
\item[(ii)] Given $\{V_j\}_{j\in J}\subset\mathbb C^{n\times n}$ finite there exists a type~I Stinespring curve $\Phi$
%$(\Phi_t)_{t\geq 0}$ 
%into $\textsc{cptp}(n)$ 
such that \eqref{eq:stinespring_second_der} holds.
%$\ddot\Phi(0)=-\sum_{j\in J}{\bf\Gamma}_{V_j}\,$.
Moreover the ancilla Hilbert space can be chosen finite-dimensional.
\end{itemize}
\end{thm}
\begin{proof}[Proof idea]
We only sketch the ideas; the full proof is given in Appendix~\appref{C}.

(i): Like in the proof of Theorem~\ref{thm_main} one computes that the second derivative of $\Phi_t$ at zero equals 
%\begin{align*}
$
-\operatorname{tr}_\omega(H^2)(\cdot) -(\cdot)\operatorname{tr}_\omega(H^2)+
2 \operatorname{tr}_{\mathcal K}\big( H  ((\cdot)\otimes\omega) H \big)
$.
%\end{align*}
This translates to $\ddot\Phi_0=-\sum_{(j,k)\in J\times N}{\bf\Gamma}_{V_{jk}}$ with
$V_{jk}:=\sqrt{2r_k}\operatorname{tr}_{|g_k\rangle\langle g_j|}(H)$ for all $j\in J$, $k\in N$
which, as a side note, resembles the Kraus operators of type~I Stinespring curves, cf.~\cite[Eq.~(3.42)]{BreuPetr02}.
Here $\sum_{k\in N}r_k|g_k\rangle\langle g_k|$ is any decomposition of the ancilla state $\omega$ with $N\subseteq\mathbb N$, $r_k>0$, and $\{g_k\}_{n\in\mathbb N}$ is an orthonormal system in $\mathcal K$ which we complete to an orthonormal basis $\{g_j\}_{j\in J}$ of $\mathcal K$.
In particular $\ddot\Phi_0$ can be written as the sum of at most $|J\times N|=\operatorname{rk}(\omega)\operatorname{dim}(\mathcal K)$ dissipative terms.
On the other hand the above expression for $\ddot\Phi_0$ yields $\Psi\in\mathcal L(\mathbb C^{n\times n})$ completely positive such that 
$\ddot\Phi_0=\Psi-\frac12\Psi^*(\mathbbm1)(\cdot)-(\cdot)\frac12\Psi^*(\mathbbm1)$.
In particular $\Psi$ admits Kraus operators $\{V_j\}_{j=1}^\ell$, $\ell\leq n^2$ which implies $\ddot\Phi_0=-\sum_{j=1}^\ell{\bf\Gamma}_{V_j}$. Thus $\ddot\Phi_0$ can always be decomposed into at most $n^2$ dissipative terms.

(ii): We adapt the strategy commonly used to convert the Kraus representation into a Stinespring dilation,
i.e.~collecting all $\{V_j\}_{j\in J}$ into one ``column'' of a bigger matrix (cf.~also Appendix~\appref{A}).
More precisely define
\begin{equation}\label{eq:H_construction}
H:=\frac{1}{\sqrt2}\sum_{j=1}^{m-1}\big(V_j\otimes|e_{j+1}\rangle\langle e_1| + V_j^*\otimes|e_{1}\rangle\langle e_{j+1}|\big)\in\mathbb C^{n\times n}\otimes\mathbb C^{m\times m}
\end{equation}
where $m:=|J|+1$.
Obviously $H$ is a Hermitian matrix so -- in addition to setting $\omega:=|e_1\rangle\langle e_1|$ -- it generates a type~I Stinespring curve. 
Now a straightforward computation shows that this curve satisfies \eqref{eq:stinespring_second_der}.
\end{proof}

Another way to view this result is that generators of quantum-dynamical semigroups can be characterized via type~I Stinespring curves (resp.~their second derivative at zero). 
After all, one needs not all differentiable but only Stinespring curves to approximate any type of Markovian dynamics close to the identity.
As a corollary, combining (the proofs of) our two main theorems one for all Hamiltonians $H\in\mathcal B(\mathbb C^n\otimes\mathcal K)$ and all states $\omega=\sum_{k\in N}r_k|g_k\rangle\langle g_k|$ on $\mathcal K$ finds that
\begin{equation}\label{eq:markov_taylor}
\begin{split}
\operatorname{tr}_{\mathcal K}\big( e^{iHt}(&(\cdot)\otimes\omega)e^{-iHt}\big)\equiv\\
&\equiv\operatorname{id}{}+ it[\operatorname{tr}_\omega(H),\,\cdot\,]-\frac{t^2}{2}\sum_{(j,k)\in J\times N}{\bf\Gamma}_{\sqrt{2r_k}\operatorname{tr}_{|g_k\rangle\langle g_j|}(H)}+\mathcal O(t^3)\,.
\end{split}
\end{equation}
Specifically, given $H_0\in\mathbb C^{n\times n}$ Hermitian and $\{V_j\}_{j=1}^m\subset\mathbb C^{n\times n}$ for some $m\in\mathbb N$, subtracting $H_0\otimes|e_1\rangle\langle e_1|$ from \eqref{eq:H_construction} yields a type~I Stinespring curve $(\Phi_t)_{t\geq 0}$ with first derivative $-i\operatorname{ad}_{H_0}$ and second derivative $-\sum_{j=1}^m{ \bf\Gamma}_{V_j}$ at zero, i.e.
$$
\Phi_t\equiv\operatorname{id}{}-it[H_0,\,\cdot\,]-\frac{t^2}{2}\sum_{j=1}^m{\bf\Gamma}_{V_j}+\mathcal O(t^3)\,.
$$
Finally, it should come at no surprise that this construction is highly non-unique:
given a Hamiltonian $H$ and a set of generating $V_j$'s there can in fact exist uncountably many different Stinespring curves which have $-i\operatorname{ad}_{H_0}$ as their first and $-\sum_j{\bf\Gamma}_{V_j}$ as their second derivative.
This is best illustrated by means of an easy example:

%........, cf.~also \cite{OSID22}. In particular this shows how the set of \textsc{gksl}-generators arise from Stinespring curves.
\section{A Qubit Example}\label{sec_example}

Let us consider a simple qubit dephasing process
$\Phi_t:=e^{tL}$ generated by a single Lindblad operator $V=\operatorname{diag}(0,1)$ via
\eqref{eq:lindblad_V}, that is, $L=-{\bf\Gamma}_V$
and
$$
\Phi_t(A)=
\Phi_t\begin{pmatrix}
a_{11}&a_{12}\\
a_{21}&a_{22}
\end{pmatrix}=\begin{pmatrix}
a_{11}&e^{-\frac{t}2}a_{12}\\
e^{-\frac{t}2}a_{21}&a_{22}
\end{pmatrix}\,.
$$
Our goal now is to find type~I Stinespring curves the second derivative of which reproduces the generator of $\Phi_t$.
%Doing so would also show non-uniqueness of our construction from the proof of Theorem~\ref{thm_main_2}.
We start with the following Ansatz:
\begin{align*}
\Psi:[0,\infty)&\to\textsc{cptp}(n)\\
A&\mapsto\operatorname{tr}_{\mathbb C^m}\big(e^{iHt}(A\otimes |e_1\rangle\langle e_1|)e^{-iHt}\big)\,,
\end{align*}
that is, the finite-dimensional ancilla is in the pure state $|e_1\rangle\langle e_1|$, but $m\in\mathbb N$ as well as the Hamiltonian $H\in\mathbb C^{2m\times 2m}$ are arbitrary for now.
The proof of Theorem~\ref{thm_main_2} shows that $\ddot\Psi(0)=-\sum_{j=1}^m {\bf\Gamma}_{V_j}$ with $V_j=\sqrt{2}\operatorname{tr}_{|e_1\rangle\langle e_j|}(H)$ for all $j=1,\ldots,m$ \footnote{
Actually one could replace $e_j$ by any $g_j$ as long as $\{g_j\}_{j=1}^m$ is an orthonormal basis of $\mathbb C^m$ with $g_1=e_1$. However, we stick to the standard basis for simplicity.
}.
Because we only have one Lindblad operator $V$ which is even Hermitian
the simplest choice is $m=1$ and $H=\frac{1}{\sqrt2}V$ so
\begin{equation}\label{eq:psi_t_1}
\Psi_t\begin{pmatrix}
a_{11}&a_{12}\\
a_{21}&a_{22}
\end{pmatrix}=e^{\frac{iVt}{\sqrt2}}\begin{pmatrix}
a_{11}&a_{12}\\
a_{21}&a_{22}
\end{pmatrix}e^{-\frac{iVt}{\sqrt2}}=\begin{pmatrix}
a_{11}&e^{-\frac{it}{\sqrt2}}a_{12}\\
e^{\frac{it}{\sqrt2}}a_{21}&a_{22}
\end{pmatrix}\,.
\end{equation}
This curve has the desired properties because
$$
\ddot\Psi_0\begin{pmatrix}
a_{11}&a_{12}\\
a_{21}&a_{22}
\end{pmatrix}=\begin{pmatrix}
0&-\frac12a_{12}\\
-\frac12a_{21}&0
\end{pmatrix}=\ddot\Phi_0\begin{pmatrix}
a_{11}&a_{12}\\
a_{21}&a_{22}
\end{pmatrix}\,.
$$
However, the first derivative of $\Psi_t$ does not yet match the Hamiltonian part of our initial $L$.
We can fix this by setting $m=2$;
%Another curve our construction yields emerges once we choose $m=2$.
then we get two dissipative terms generated by $V_1=\sqrt{2}\operatorname{tr}_{|e_1\rangle\langle e_1|}(H)$ and $V_2=\sqrt{2}\operatorname{tr}_{|e_1\rangle\langle e_2|}(H)$.
We can either choose $V_1=V$ and $V_2=0$ (which again yields \eqref{eq:psi_t_1}) or we can choose $V_1=0$, $V_2=V$.
A direct computation shows that for the latter case $H$ has to be of the form
\begin{equation}\label{eq:ex_H}
H=\begin{pmatrix}
0&0&0&0\\
0&a&0&b\\
0&0&0&\frac{1}{\sqrt{2}}\\
0&b^*&\frac{1}{\sqrt{2}}&c
\end{pmatrix}
\end{equation}
where $a,c\in\mathbb R$, $b\in\mathbb C$ can be arbitrary. These parameters may yield different overall curves, but they do not affect the first and second derivative of $\Psi_t$ at zero.
For example choosing $a=b=c=0$ yields
\begin{align}
\Psi_t(A)
&=\operatorname{tr}_{\mathcal K}\begin{pmatrix}
a_{11}&0&\cos(\frac{t}{\sqrt2})a_{12}&-i\sin(\frac{t}{\sqrt2})a_{12}\\
0&0&0&0\\
\cos(\frac{t}{\sqrt2})a_{21}&0&(\cos(\frac{t}{\sqrt2}))^2a_{22}&-\frac{i}2\sin(\sqrt2t)a_{22}\\
i\sin(\frac{t}{\sqrt2})a_{21}&0&\frac{i}2\sin(\sqrt2t)a_{22}&(\sin(\frac{t}{\sqrt2}))^2a_{22}
\end{pmatrix}\nonumber\\
&=\begin{pmatrix}
a_{11}&a_{12}\cos(\frac{t}{\sqrt2})\\
a_{21}\cos(\frac{t}{\sqrt2})&a_{22}
\end{pmatrix}\label{eq:psi_t_2}
\end{align}
which obviously has the correct first and second derivative at zero.
At this point, a few remarks are in order:
\begin{itemize}
\item It is quite easy to see that this particular $\Psi$ cannot be of \textsc{gksl}-form simply because $\Psi_t$ fails to be bijective whenever $t=\frac{\pi}{\sqrt2}+\sqrt2\pi k$ for some $k\in\mathbb N_0$. Moreover $\Psi_t$ is \textsc{p}-divisible (cf.~\cite{Plenio_NMarkov_Review2014}) for $0\leq t\leq\frac{\pi}{\sqrt2}$ but fails to be \textsc{p}-divisible for $\frac{\pi}{\sqrt2}<t<\sqrt2\pi$, etc.~which further showcases its non-Markovian behavior, refer also to \cite{ChruKossRivas11}.
\item The previous observation is an incarnation of the quantum recurrence theorem \cite{Bocchieri57,Schulman78} (cf.~also \cite{wallace2015recurrence,keyl18InfLie}): a one-parameter unitary group in finite dimensions revisits any given point in time either exactly or at least norm-approximately.
%recurrence: In finite dimension a one parameter unitary group exp(itK) with selfadjoint generator K ∈ B(H) always revisits
%its own past – either exactly (if the group is periodic) or at least approximately. The latter means that for all
%t− < 0 and all ǫ > 0 there is a t+ > 0 with
This is another hint as to why Stinespring curves with finite-dimensional environments cannot model Markovian open system interactions: because 
$e^{iHt}$ -- and thus the induced Stinespring curve -- eventually revisits the identity,
the dissipative effects have to be reversed at some point. But this violates Markovianity. 
\item From \eqref{eq:ex_H} 
%, $\dot\Psi_0\equiv 0$ (due to $\operatorname{tr}_{|e_1\rangle\langle e_1|}(H)=0$) whereas the first derivative of \eqref{eq:psi_t_1} does not vanish at zero.
%The latter is not necessary as 
we can easily construct another curve which differs from \eqref{eq:psi_t_2} only in the third (and any higher) derivative. For this choose $a=c=0$ and $b=\frac{1}{\sqrt2}$ in \eqref{eq:ex_H}; then $\Psi_t(A)$ equals
\begin{align*}
\begin{pmatrix}
a_{11}+a_{22}(\sin(\frac{t}{2}))^4&a_{12}(\cos(\frac{t}{2}))^2 -\frac{ i\sin(t)(\cos(t)-1) }{2\sqrt2} a_{22} \\
a_{21}(\cos(\frac{t}{2}))^2 +\frac{ i\sin(t)(\cos(t)-1) }{2\sqrt2}a_{22}&a_{22}(1-(\sin(\frac{t}{2})^4)
\end{pmatrix}\,,
\end{align*}
but
$$
\dddot\Psi_0(A)=\begin{pmatrix}
0&\frac{3i}{2\sqrt2}a_{22}\\
-\frac{3i}{2\sqrt2}a_{22}&0 \end{pmatrix}
$$
while the third derivative of \eqref{eq:psi_t_2} vanishes at zero.
This also showcases non-uniqueness of our construction: each value of $b$ yields a different curve\footnote{
For $a=c=0$ and general $b\in\mathbb R$ one finds $\dddot\Psi_0(A)=-\frac{3}{2}b\sigma_y\cdot a_{22}$.
}
which nonetheless has the same first and second derivative at zero.
\end{itemize}
%so $\dot\Phi(0)\equiv 0$ which -- using \eqref{eq:partial_trace_state} -- matches
%$$
% \dot\Phi(0)=i\operatorname{tr}_{|e_1\rangle\langle e_1|}(H)=i\begin{pmatrix}
%\langle e_1\otimes e_1,H(e_1\otimes e_1)\rangle&\langle e_1\otimes e_1,H(e_2\otimes e_1)\rangle\\
%\langle e_2\otimes e_1,H(e_1\otimes e_1)\rangle&\langle e_2\otimes e_1,H(e_2\otimes e_1)\rangle
%\end{pmatrix}=0,
%$$
%and
%$$
%\ddot\Phi(0)(A)=\begin{pmatrix}
%0&-\frac12x_{12}\\
%-\frac12a_{21}&0
%\end{pmatrix}
%$$
%which reproduces
%%the action of 
%$-{\bf\Gamma}_V\,$.
\section{Conclusions and Outlook}
We investigated how a dynamical version of Stinespring's dilation theorem interacts with the notion of quantum-dynamical semigroups.
In doing so we found that if a Stinespring curve which models non-trivial environment interaction at the same time is (time-independent) Markovian,
then the dynamics of system plus environment have to be generated by an unbounded---and thus infinite-dimensional---Hamiltonian.
This complements known results about how bounded Hamiltonians can only generate sub-exponential decay.
Moreover we gave an explicit -- albeit non-unique -- construction for converting the dissipative part of a \textsc{qds} into (the second derivative of) a Stinespring curve, and we showed that every \textsc{qds} arises this way.

From here there are two rather obvious directions to pursue:
First, based on our improved understanding of how \textsc{qds} and Markovianity interact one may attempt to generalize Davies' existence result \cite{Davies72} from finite dimensions to norm-continuous (or even strongly continuous) \textsc{qds} in infinite dimensions.
Second, Stinespring's dilation theorem as well as the notion of Stinespring curves
are essential for quantum thermodynamics: they are used to define the so-called thermal operations which are the fundamental building block of the resource theory approach to quantum thermodynamics \cite{Lostaglio19,vomEnde22thermal}.
While studying the intersection between Markovianity and quantum thermodynamics is a quite recent field \cite{LosKor22a}
our results suggest that already the definition of thermal operations
holds valuable insights in this direction.
Indeed we will explore this in future work \cite{OSID22}.

\section*{Acknowledgments}
I would like to thank
%Emanuel Malvetti, 
Gunther Dirr
%, Thomas Schulte-Herbr\"uggen, and Amit Devra
%and the anonymous referee
for valuable and constructive comments during the preparation of this manuscript,
as well as Daniel Burgarth for 
drawing my attention to recent publications on dilations of quantum-dynamical semigroups.
%Moreover I am grateful to the anonymous referee for their valuable comments which led to an improved presentation of the material.
This research is part of the Bavarian excellence network \textsc{enb}
via the International PhD Programme of Excellence
\textit{Exploring Quantum Matter} (\textsc{exqm}), as well as the \textit{Munich Quantum Valley} of the Bavarian
State Government with funds from Hightech Agenda \textit{Bayern Plus}.

\section*{Appendix \app{A}: Relation Between Choi-Jamio\l{}kowski, Kraus, and Stinespring}

Given $m,n\in\mathbb N$ and $\Phi:\mathbb C^{n\times n}\to\mathbb C^{m\times m}$ linear we say $\Phi$ is completely positive if $\Phi\otimes\operatorname{id}_k$ (or $\operatorname{id}_k\otimes\Phi$) for all $k\in\mathbb N$ maps positive semi-definite matrices to positive semi-definite matrices. Equivalently, the Choi matrix
$$
C(\Phi)=\sum_{j,k=1}^n|e_j\rangle\langle e_k|\otimes\Phi(|e_j\rangle\langle e_k|)\in\mathbb C^{mn\times mn}$$
is positive semi-definite \cite[Thm.~1]{Choi75}, and Choi's proof explicitly constructs a set of Kraus operators:
first decompose $C(\Phi)=\sum_i|\psi_i\rangle\langle\psi_i|$ into pairwise orthogonal vectors $\{\psi_i\}_i\subset\mathbb C^n\otimes\mathbb C^m$, then set $K_i:=\operatorname{vec}^{-1}(\psi_i)$ where
\begin{align*}
\operatorname{vec}:\mathbb C^{m\times n}&\to\mathbb C^n\otimes\mathbb C^m\\
X&\mapsto\sum_{i=1}^ne_i\otimes Xe_i
\end{align*}
is the vectorization isomorphism \cite[Ch.~2.4]{MN07} and its inverse is given by $\operatorname{vec}^{-1}(\psi)=(\langle e_j\otimes e_i,\psi\rangle)_{i=1,j=1}^{m,n}$.
This also shows that if, conversely, $\{K_i\}_{i}$ is an arbitrary set of Kraus operators of $\Phi$, i.e.~$\Phi\equiv\sum_iK_i(\cdot)K_i^*$, then $C(\Phi)=\sum_i\operatorname{vec}(K_i)\operatorname{vec}(K_i)^*$.\medskip

Now for the Stinespring dilation. In its original form (assuming finite-dimensions and working in the Schr\"odinger picture), given a linear, completely positive map $\Phi:\mathbb C^{n\times n}\to\mathbb C^{m\times m}$
it states that there exists $\ell\in\mathbb N$ and $U_0:\mathbb C^m\to\mathbb C^n\otimes\mathbb C^\ell$ (that is, $U_0\in\mathbb C^{n\ell\times m}$) such that \cite[Thm.~6.9]{Holevo12}
\begin{equation}\label{eq:stinespring_1}
\Phi\equiv\operatorname{tr}_{\mathbb C^\ell}\big(U_0(\cdot)U_0^*)
\end{equation}
with $\operatorname{tr}_{\mathbb C^\ell}:\mathbb C^{n\times n}\otimes\mathbb C^{\ell\times\ell}\to\mathbb C^{n\times n}$ being the usual partial trace.
Starting from a set of Kraus operators $\{K_i\}_{i=1}^\ell \subset\mathbb C^{m\times n}$ of $\Phi$ define
$U_0:\mathbb C^m\to\mathbb C^n\otimes\mathbb C^\ell$ via
$U_0x:=\sum_{i=1}^\ell (K_ix)\otimes e_i$ \cite[Eq.~(6.20) ff.]{Holevo12} because then
\begin{align*}
\operatorname{tr}_{\mathbb C^\ell}\big(U_0|x\rangle\langle y|U_0^*)&=
\operatorname{tr}_{\mathbb C^\ell}\big(|U_0x\rangle\langle U_0y|)=\sum_{i,j=1}^\ell\operatorname{tr}_{\mathbb C^\ell}\big( K_i|x\rangle\langle y|K_j^*\otimes |e_i\rangle\langle e_j| )\\
&=\sum_{i,j=1}^\ell\langle e_j,e_i\rangle K_i|x\rangle\langle y|K_j^*=\sum_{i=1}^\ell K_i|x\rangle\langle y|K_i^*=\Phi(|x\rangle\langle y|)
\end{align*}
for all $x,y\in\mathbb C^n$.
This shows \eqref{eq:stinespring_1} because the rank-1 operators span $\mathbb C^{n\times n}$.
Conversely if $\Phi$ is of the form \eqref{eq:stinespring_1}, then $\{K_i\}_{i=1}^\ell$ is a set of Kraus operators of $\Phi$ where $K_iy:=\sum_{j=1}^n\langle e_j\otimes e_i,U_0y\rangle e_j\in\mathbb C^m$ for all $y\in\mathbb C^n$, $i=1,\ldots,\ell$.

Now this theorem takes a special form if $\Phi$ additionally is trace-preserving;
this is precisely when the operator $U_0$ in \eqref{eq:stinespring_1} is an isometry \cite[Thm.~6.9]{Holevo12}.
In this case $\Phi$ ``can be extended to the evolution of an open system interacting with an environment'' \cite[Thm.~6.18]{Holevo12}.
More precisely let $\Phi:\mathbb C^{n\times n}\to\mathbb C^{m\times m}$ be completely positive and trace-preserving, and let $d:=\operatorname{lcm}(m,n)$ be the least common multiple of $m$ and $n$. Then this version of Stinespring's theorem (cf.~\cite[Coro.~1]{vE_dirr_semigroups}) guarantees the existence of a number $\ell\in\mathbb N$ and a unitary matrix $U\in\mathbb C^{d\ell\times d\ell}$
%, and a vector $\psi\in\mathbb C^{d\ell /n}$, $\|\psi\|=1$ 
such that
\begin{equation}\label{eq:Stinespring_cptp}
\Phi\equiv\operatorname{tr}_{\mathbb C^{d\ell/m}}\big(U((\cdot)\otimes|e_1\rangle\langle e_1|)U^*\big)\,.
\end{equation}
Note that if $m=n$, then $d=m=n$ so this result reproduces the common formulation of Stinespring's theorem for quantum maps \cite[Thm.~6.18]{Holevo12}.
As for the construction:
starting again from a set of Kraus operators $\{K_i\}_{i=1}^{\ell'}$ of $\Phi$, constructing $U$ amounts to collecting the $K_i$ ``in the first column'' of a larger matrix and filling up the rest of $U$ such that it becomes unitary.
More precisely \cite[Eq.~(6.23)]{Holevo12} we define $\ell:=\ell'\cdot\frac{d}{n}$
%, $\psi:=e_1\otimes e_1\in\mathbb C^{\frac{d}{n}}\otimes\mathbb C^\ell$, 
and
\begin{align*}
U_0:= \sum_{i=1}^{\ell'}\sum_{j=1}^{\frac{d}{n}} &K_i  \otimes|e_1\rangle\langle e_j|\otimes|e_i\rangle\langle e_1|\otimes|e_j\rangle\langle e_1|\\
\in\;&\mathbb C^{m\times n} \otimes\mathbb C^{\frac{d}{m}\times\frac{d}{n}}\otimes\mathbb C^{\ell'\times\ell'}\otimes\mathbb C^{\frac{d}{n}\times\frac{d}{n}}
\simeq
\mathbb C^{d\ell\times d\ell}
%\simeq\mathbb C^{n\times n}\otimes\mathbb C^{\frac{d\ell}{n}\times\frac{d\ell}{n}}
\,,
\end{align*}
and notice that
\begin{align*}
U_0^*U_0&=\sum_{i=1}^{\ell'}\sum_{j=1}^{\frac{d}{n}}K_i^*K_i\otimes|e_j\rangle\langle e_j|\otimes|e_1\rangle\langle e_1|\otimes|e_1\rangle\langle e_1|\\
&=\Big(\sum_{i=1}^{\ell'}K_i^*K_i\Big)\otimes  \mathbbm1_{\frac{d}{n}}\otimes|e_1\rangle\langle e_1|=\mathbbm 1_d\otimes|e_1\rangle\langle e_1|\in\mathbb C^{d\times d}\otimes\mathbb C^{\ell\times\ell}
\end{align*}
as the Kraus operators of $\Phi$ satisfy $\sum_{i=1}^{\ell}K_i^*K_i=\mathbbm 1_n$ due to trace-preservation \cite[Coro.~6.13]{Holevo12}.
In other words the ``first $d$ columns'' of $U_0$ form an orthonormal system in $\mathbb C^{d\ell}$ which can be completed to an orthonormal basis of $\mathbb C^{d\ell}$. We fill up the remaining ``columns'' of $U_0$ with these additional vectors to obtain a unitary matrix\footnote{
More precisely we arrange the new vectors into matrices $U_{ii'jj'}\in\mathbb C^{d\times d}$ with indices $i,i'=1,\ldots,\ell'$, $j,j'=1,\ldots,\frac{d}{n}$ -- but $(i',j')\neq(1,1)$ -- such that
\begin{align*}
U= \sum_{i=1}^{\ell'}\sum_{j=1}^{\frac{d}{n}} K_i  \otimes|e_1\rangle\langle e_j|\otimes|e_i\rangle\langle e_1|\otimes|e_j\rangle\langle e_1|+ \sum_{i,i'=1}^{\ell'}\sum_{\substack{j,j'=1\\(i',j')\neq(1,1)}}^{\frac{d}{n}} U_{ii'jj'}  \otimes|e_i\rangle\langle e_{i'}|\otimes|e_j\rangle\langle e_{j'}|
\end{align*}
is unitary.
} $U\in\mathbb C^{d\ell\times d\ell}$.
Now \eqref{eq:Stinespring_cptp} holds because for all $A\in\mathbb C^{n\times n}$, $B\in\mathbb C^{m\times m}$
\begin{align*}
\operatorname{tr}\big(B &\operatorname{tr}_{\mathbb C^{d\ell/m}}\big(U(A\otimes|e_1\rangle\langle e_1|)U^*\big) \big)\\
&=\operatorname{tr}\big((B\otimes\mathbbm1_{\frac{d}{m}}\otimes\mathbbm1_\ell) U(A\otimes|e_1\rangle\langle e_1|\otimes|e_1\rangle\langle e_1|)U^* \big)\\
&=\sum_{i=1}^{\ell'}\sum_{j=1}^{\frac{d}{n}}\operatorname{tr}\big( (B\otimes\mathbbm1_{\frac{d}{m}})(K_i\otimes|e_1\rangle\langle e_j|)(A\otimes|e_1\rangle\langle e_1|)(K_{i}^*\otimes|e_{j}\rangle\langle e_1|) \big)
%\langle e_{i'},e_i\rangle\langle e_{j'},e_j\rangle
\\
&= \sum_{i=1}^{\ell'}\operatorname{tr}(BK_iAK_i^*)\sum_{j=1}^{\frac{d}{n}}\langle e_j,e_1\rangle=\operatorname{tr}(B\Phi(A))\,.
\end{align*}
This procedure is reversible, as well: starting from \eqref{eq:Stinespring_cptp}, first
%complete $\psi$ to an orthonormal basis $\{g_i\}_{i=1}^{d\ell/n}$ of $\mathbb C^{d\ell/n}$ (i.e.~$g_1=\psi$) and 
decompose $U$ into ``blocks'' via
$
U=\sum_{i,j=1}^{\ell}U_{ij}\otimes|e_i\rangle\langle e_j|
$
with $U_{ij}\in\mathbb C^{d\times d}$.
Then a set of Kraus operators of $\Phi$ is given by $\{K_{ij}\}_{i=1,j=1}^{(d/m),\ell}\subset\mathbb C^{m\times n}$ where
$K_{ij}y:=\sum_{k=1}^m\langle e_k\otimes e_i,U_{j1}(y\otimes e_1)\rangle e_k
$
for all $i=1,\ldots,\frac{d}{m}$, $j=1,\ldots,\ell$, and all $y\in\mathbb C^n$.
\section*{Appendix \app{B}: Proof of Theorem~\ref{thm_main}}
First we need a result about derivatives of generalized Stinespring curves:
\begin{lemma}\label{lemma_der_stinespring_general}
Given any complex Hilbert space $\mathcal K$, any operator $B\in\mathcal B(\mathbb C^n\otimes\mathcal K)$, and any bounded linear map $E:\mathbb C^{n\times n}\to\mathcal B^1(\mathbb C^n\otimes\mathcal K)$
one has
\begin{align*}
\frac{d}{dt}\operatorname{tr}_{\mathcal K}\big( &e^{Bt} E(\cdot)e^{-Bt} \big)
=\operatorname{tr}_{\mathcal K}\big( e^{Bt} B E(\cdot)e^{-Bt} \big)+\operatorname{tr}_{\mathcal K}\big( e^{Bt} E(\cdot)(-B)e^{-Bt} \big)\,.
\end{align*}
\end{lemma}
\begin{proof}
To ease notation we introduce the bilinear map
\begin{align*}
\Lambda:\mathcal B(\mathbb C^n\otimes\mathcal K)\times\mathcal B(\mathbb C^n\otimes\mathcal K)&\to\mathcal (\mathbb C^{n\times n}\to\mathbb C^{n\times n})\\
(B_1,B_2)&\mapsto \operatorname{tr}_{\mathcal K}\big(B_1E(\cdot)B_2\big)\,.
\end{align*}
%Re-writing $\operatorname{id}\equiv\operatorname{tr}_{\mathcal K}((\cdot)\otimes\omega)$
so we have to show that
\begin{align*}
%&\lim_{t\to 0^+}\Big\|\frac{\Phi_t-\Phi_0}{t}-  \operatorname{tr}_{\mathcal K}(B((\cdot)\otimes\omega))
%-\operatorname{tr}_{\mathcal K}(((\cdot)\otimes\omega)-B) \Big\|_{1\to 1}
%\\
%&\qquad\quad=
\lim_{h\to 0}\Big\|  &\frac{\Lambda(e^{B(t+h)},e^{-B(t+h)})-\Lambda(e^{Bt},e^{-Bt})}{h} \\
&\qquad \quad \qquad -\Lambda(e^{Bt}B,e^{-Bt})-\Lambda(e^{Bt},-Be^{-Bt})  \Big\|_{1\to 1}=0\,.
\end{align*}
The expression inside the norm can be re-written as
\begin{align}
&\tfrac{\Lambda(e^{B(t+h)},e^{-B(t+h)})\mp\Lambda(e^{Bt},e^{-B(t+h)})-\Lambda(e^{Bt},e^{-Bt})}{h\nonumber}\\
&\quad \pm \Lambda( e^{Bt}B , e^{-B(t+h)} )-\Lambda(e^{Bt}B,e^{-Bt})-\Lambda(e^{Bt},-Be^{-Bt}) \nonumber\\
&\ =\frac{\Lambda(e^{B(t+h)},e^{-B(t+h)})-\Lambda(e^{Bt},e^{-B(t+h)})}{h}- \Lambda( e^{Bt}B , e^{-B(t+h)} )\label{lemma_diff_eq:1}\\
&\ +\Lambda( e^{Bt}B , e^{-B(t+h)} )- \Lambda(e^{Bt}B,e^{-Bt}) \label{lemma_diff_eq:2}\\
&\ +\frac{\Lambda(e^{Bt},e^{-B(t+h)})-\Lambda(e^{Bt},e^{-Bt})}{h}-\Lambda(e^{Bt},-Be^{-Bt})\,.\label{lemma_diff_eq:3}
\end{align}
Bounding each of these three differences is straightforward: 
for \eqref{lemma_diff_eq:1} one finds
\begin{align*}
&\Big\|\frac{\Lambda(e^{B(t+h)},e^{-B(t+h)})-\Lambda(e^{Bt},e^{-B(t+h)})}{h}- \Lambda( e^{Bt}B , e^{-B(t+h)} )\Big\|_{1\to 1}\\
&\qquad\qquad\qquad\qquad=\Big\|\Lambda\Big(\frac{e^{B(t+h)}-e^{Bt}}h- e^{Bt}B, e^{-B(t+h)} \Big)\Big\|_{1\to 1}\\
&\qquad\qquad\qquad\qquad\leq\Big\|\Big(\frac{e^{B(t+h)}-e^{Bt}}h- e^{Bt}B\Big)E(\cdot)e^{-B(t+h)}\Big\|_{1\to 1}\\
&\qquad\qquad\qquad\qquad\leq \Big\|\frac{e^{B(t+h)}-e^{Bt}}h- e^{Bt}B\Big\|_\infty\|E\|_{1\to 1}\|e^{-B(t+h)}\|_\infty\\
%=\Big\|\frac{e^{Bt}-\mathbbm1}t-B\Big\|_\infty
&\qquad\qquad\qquad\qquad\leq\|e^{Bt}\|_\infty\Big\|\frac{e^{Bh}-\mathbbm1}{h}-B\Big\|_\infty\|E\|_{1\to 1}\|e^{-Bt}\|_\infty\|e^{-Bh}\|_\infty
\,.
\end{align*}
In the third line we used that every \textsc{cptp} map between trace classes -- which includes the partial trace -- has operator norm one. This is a well-known consequence of the Russo-Dye theorem \cite{Russo66}, cf.~also \cite[Prop.~2]{vE_dirr_semigroups}.
Then in the fourth and fifth line we used that $\|XYZ\|_1\leq\|X\|_\infty\|Y\|_1\|Z\|_\infty$ for suitable objects $X,Y,Z$ \cite[Lemma 16.6.6]{MeiseVogt97en}.
%, as well as $\|e^{itB}\|_\infty=1$ for all $t\in\mathbb R$ because $B$ is self-adjoint so $e^{itB}$ is unitary.
%as well as unitarity of $e^{-Bt}$.
%In total this shows that \eqref{lemma_diff_eq:1} is upper bounded in norm by $\|\frac{e^{Bt}-\mathbbm1}t-B\|_\infty$.
Similarly,
%, one sees that
%\begin{itemize}
%\item 
\eqref{lemma_diff_eq:3} is upper bounded in norm by $\|e^{Bt}\|_\infty\|E\|_{1\to 1}\|\frac{e^{-Bh}-\mathbbm1}{h}+B\|_\infty\|e^{-Bt}\|_\infty\,$,
%\item
and for
\eqref{lemma_diff_eq:2}
%is upper bounded in norm by
we compute
\begin{align*}
\big\|\Lambda( e^{Bt}B , e^{-B(t+h)} )&- \Lambda(e^{Bt}B,e^{-Bt}) \big\|_{1\to 1}\\
%&=\big\| \operatorname{tr}_{\mathcal K}\big(B\big((\cdot)\otimes\omega\big)(e^{-Bt}-\mathbbm1)\big)\big\|_{1\to 1}\\
&\leq \|e^{Bt}\|_\infty\|B\|_\infty\|E\|_{1\to 1}\|e^{-Bt}\|_\infty\|e^{-Bh}-\mathbbm1\|_\infty\,.
\end{align*}
Combining these estimates yields
\begin{align*}
\Big\|  &\frac{\Lambda(e^{B(t+h)},e^{-B(t+h)})-\Lambda(e^{Bt},e^{-Bt})}{h} \\
&\qquad \quad \qquad\qquad  -\Lambda(e^{Bt}B,e^{-Bt})-\Lambda(e^{Bt},-Be^{-Bt})  \Big\|_{1\to 1}\\
&\leq\Big( \Big\|\frac{e^{Bh}-\mathbbm1}{h}-B\Big\|_\infty\|e^{-Bh}\|_\infty+\Big\|\frac{e^{-Bh}-\mathbbm1}{h}+B\Big\|_\infty\\
&\qquad \quad \qquad\qquad +\|B\|_\infty\|e^{-Bh}-\mathbbm1\|_\infty \Big) \|e^{Bt}\|_\infty\|e^{-Bt}\|_\infty\|E\|_{1\to 1}\,.
\end{align*}
But this expression vanishes as $h\to 0$ because $E$ is bounded and because
$\frac{d}{dt}e^{\pm Bt}|_{t=0}=\pm B$ in norm due to $B$ being bounded.
\end{proof}
Now computing the derivative of $(\Phi_t)_{t\geq 0}$ at zero is a mere application of Lemma~\ref{lemma_der_stinespring_general}: choosing
$B=iH$ and $E(\cdot)=(\cdot)\otimes\omega$
we find that $\dot\Phi_0$
given by the map
$X\mapsto \operatorname{tr}_{\mathcal K}\big(iH(X\otimes\omega)\big)
+\operatorname{tr}_{\mathcal K}\big((X\otimes\omega)(-iH)\big)$ on
$\mathbb C^{n\times n}$. Note that this map is well defined because 
$X\otimes\omega$ is in $\mathbb C^{n\times n}\otimes\mathcal B^1(\mathcal K)
\simeq\mathcal B^1(\mathbb C^n\otimes\mathcal K)$ \cite[p.~34]{Kraus83}, hence 
the argument of $\operatorname{tr}_{\mathcal K}:\mathcal B^1(
\mathbb C^n\otimes\mathcal K)\to\mathcal B^1(\mathbb C^n)\simeq
\mathbb C^{n\times n}$ is trace class, as well.
All that is left is to cast the expression we obtained into a more familiar form:
\begin{lemma}\label{lemma_rewrite_trace_expression}
Given $n\in\mathbb N$, a Hilbert space $\mathcal K$, as well as $B\in\mathcal B(\mathbb C^n\otimes\mathcal K)$, $A\in\mathcal B^1(\mathcal K)$, and $X\in\mathbb C^{n\times n}$ the following statements hold.
\begin{itemize}
\item [(i)] Defining 
$\operatorname{tr}_A(B)$ as the unique $n\times n$ matrix which for all $X\in\mathbb C^{n\times n}$ satisfies
$
\operatorname{tr}( \operatorname{tr}_A(B)X )=\operatorname{tr}(B(X\otimes A)  )
$, called ``partial trace of $B$ with respect to $A$'', the map $\operatorname{tr}_A:\mathcal B(\mathbb C^n\otimes\mathcal K)\to\mathbb C^{n\times n}$, $B\mapsto\operatorname{tr}_A(B)$ is the dual of $\iota_A:\mathbb C^{n\times n}\to\mathcal B^1(\mathbb C^n\otimes\mathcal K)$, $X\mapsto X\otimes A$. Thus $\operatorname{tr}_A$ is well-defined, and is completely positive identity-preserving if and only if $A$ is a state.
\item[(ii)] $(\operatorname{tr}_A(B))^*=\operatorname{tr}_{A^*}(B^*)$
\item[(iii)]
One has
\begin{equation}\label{eq:lemma_rewr_trace_1}
\operatorname{tr}_{\mathcal K}\big(B(X\otimes A)\big)=\operatorname{tr}_ A(B)X
\end{equation}
and
\begin{equation}\label{eq:lemma_rewr_trace_2}
\operatorname{tr}_{\mathcal K}\big((X\otimes A)B\big)=X\operatorname{tr}_ A(B)\,.
\end{equation}
\item[(iv)] Given any orthonormal basis $\{g_j\}_{j=1}^n$ of $\mathbb C^n$ one has
\begin{equation*}
%\label{eq:partial_trace_state}
\operatorname{tr}_ A(X)
%=\sum_{j,k=1}^n\operatorname{tr}{}\big(X(|g_j\rangle\langle g_k|\otimes A)\big)|g_j\rangle\langle g_k|
=\sum_{j,k=1}^n\operatorname{tr}( A X_{kj})|g_j\rangle\langle g_k|
\end{equation*}
where $X_{jk}$ are the ``blocks'' of $X$, i.e.~$X=\sum_{j,k=1}^n|g_j\rangle\langle g_k|\otimes X_{jk}$.
\end{itemize}
\end{lemma}
\begin{proof}
(i): See \cite[Ch.~9, Lemma 1.1]{Davies76} or \cite[Ch.~II.B]{vE_dirr_semigroups}.
(ii): For all $X\in\mathbb C^{n\times n}$
\begin{align*}
\operatorname{tr}\big(X (\operatorname{tr}_A(B))^* \big)&=\overline{\operatorname{tr}\big(X^*\operatorname{tr}_A(B)\big)}\\
&=\overline{\operatorname{tr}((X^*\otimes A)B)}=\operatorname{tr}((X\otimes A^*)B^*)=\operatorname{tr}\big(X\operatorname{tr}_{A^*}(B^*)\big)\,.
\end{align*}
(iii). Obviously, given any $X_1,X_2\in\mathbb C^{n\times n}$ one has $X_1=X_2$ if and only if $\operatorname{tr}(\omega X_1)=\operatorname{tr}(\omega X_2)$ for all $\omega \in\mathbb C^{n\times n}$. Thus \eqref{eq:lemma_rewr_trace_1} follows from
\begin{align*}
\operatorname{tr}\big(\omega \operatorname{tr}_{\mathcal K}\big(B(X\otimes A)\big)\big)&=\operatorname{tr}\big((\omega \otimes\mathbbm1)B(X\otimes A)\big)\\
&=\operatorname{tr}\big((X\omega \otimes A)B\big)\\
&=\operatorname{tr}\big(X\omega \operatorname{tr}_ A(B))=\operatorname{tr}\big(\omega \operatorname{tr}_ A(B)X)\,,
\end{align*}
and \eqref{eq:lemma_rewr_trace_2} is proven analogously. Finally, (iv) is readily verified.
\end{proof}
Therefore
%our previous computation shows
$$
\frac{d}{dt}\Phi_t\Big|_{t=0}=\operatorname{tr}_\omega(iH)(\cdot)+(\cdot)\operatorname{tr}_{\omega}(-iH)=i\big[\operatorname{tr}_\omega(H),\,\cdot\,\big]\,.
$$
On the other hand $\Phi$ by assumption is a quantum-dynamical semigroup
so $\frac{d}{dt}\Phi_t|_{t=0}$ equals its \textsc{gksl}-generator \cite{GKS76,Lindblad76}.
Therefore
$$
\Phi_t\equiv e^{t\operatorname{ad}_{\operatorname{tr}_\omega(iH)}}= e^{it \operatorname{tr}_\omega(H)}(\cdot)e^{-it \operatorname{tr}_\omega(H)}
$$
meaning that $\Phi_t$ is a unitary channel for all $t\geq 0$.
Here we used that $\operatorname{tr}_\omega(H)$ is Hermitian because $H$ and $\omega$ are self-adjoint (Lemma~\ref{lemma_rewrite_trace_expression} (ii)).

\section*{Appendix \app{C}: Proof of Theorem~\ref{thm_main_2}}

Applying Lemma~\ref{lemma_der_stinespring_general} multiple times (with $B=iH$ as well as $E=(\cdot)\otimes\omega$, $E=iH((\cdot)\otimes\omega)$, and $E=((\cdot)\otimes\omega)(-iH)$, respectively) -- together with Lemma~\ref{lemma_rewrite_trace_expression} (iii) -- shows that the second derivative of any type~I Stinespring curve at zero is given by
\begin{align}
\operatorname{tr}_{\mathcal K}\big( (iH)^2& ((\cdot)\otimes\omega)  \big)+
\operatorname{tr}_{\mathcal K}\big(  ((\cdot)\otimes\omega) (-iH)^2 \big)+
2\operatorname{tr}_{\mathcal K}\big( (iH)  ((\cdot)\otimes\omega) (-iH) \big)\nonumber\\
&=-\operatorname{tr}_{\mathcal K}\big( H^2 ((\cdot)\otimes\omega)  \big)-
\operatorname{tr}_{\mathcal K}\big(  ((\cdot)\otimes\omega) H^2 \big)+
2\operatorname{tr}_{\mathcal K}\big( H  ((\cdot)\otimes\omega) H \big)\nonumber\\
&=-\operatorname{tr}_\omega(H^2)(\cdot) -(\cdot)\operatorname{tr}_\omega(H^2)+
2\operatorname{tr}_{\mathcal K}\big( H  ((\cdot)\otimes\omega) H \big)\label{eq:Stinespring_second_der2}
\end{align}

(i): We decompose $\omega=\sum_{k\in N}r_k|g_k\rangle\langle g_k|$ for some $N\subseteq\mathbb N$, $r_k>0$, and some orthonormal system $\{g_k\}_{n\in\mathbb N}$ in $\mathcal K$ \cite[Prop.~16.2]{MeiseVogt97en}. We complete the latter to an orthonormal basis $\{g_j\}_{j\in J}$ of $\mathcal K$  (i.e.~$N\subseteq J$) \cite[Prop.~12.6]{MeiseVogt97en}
and claim that $\ddot\Phi_0=-\sum_{(j,k)\in J\times N}{\bf\Gamma}_{V_{jk}}$ where
$V_{jk}:=\sqrt{2r_k}\operatorname{tr}_{|g_k\rangle\langle g_j|}(H)$ for all $j\in J$, $k\in N$. In order to prove this we need the following lemma regarding partial traces:

\begin{lemma}\label{lemma_partial_trace}
As in Lemma~\ref{lemma_rewrite_trace_expression} let $\operatorname{tr}_{(\cdot)}$ denote the partial trace with respect to a given trace class operator.
For all Hilbert spaces $\mathcal H,\mathcal K$
%as well as $B\in\mathcal B(\mathcal H\otimes\mathcal K)$, $X\in\mathcal B^1(\mathcal K)$, and $\psi,\phi\in\mathcal K$ 
the following hold:
\begin{itemize}
\item[(i)] For all $B\in\mathcal B(\mathcal H\otimes\mathcal K)$ and all $\psi,\phi\in\mathcal K$
one has
$\operatorname{tr}_{|\psi\rangle\langle\phi|}(B)=\iota_\phi^*B\iota_\psi$ 
and $B\otimes|\psi\rangle\langle\phi|=\iota_\psi B\iota_\phi^*$
where, here and henceforth, $\iota_y:\mathcal H\to\mathcal H\otimes\mathcal K$ for any $y\in\mathcal K$ is the map $x\mapsto x\otimes y$.
\item[(ii)] If $\{g_j\}_{j\in J}$ is any orthonormal basis of $\mathcal K$, then $\sum_{j\in J}\iota_{g_j}\iota_{g_j}^*=\mathbbm1_{\mathcal H\otimes\mathcal K}$ in the strong operator topology.
\end{itemize}
\end{lemma}
\begin{proof}
(i):
%Analogous to the proof of \cite[Lemma A.3.1]{vE_PhD_2020}.
For all $x,y\in\mathcal H$
\begin{align*}
\langle x,\iota_\phi^*B\iota_\psi y\rangle=\langle \iota_\phi x,B\iota_\psi y\rangle&=\langle (x\otimes\phi),B(y\otimes\psi)\rangle \\
&=\operatorname{tr}\big(|y\otimes\psi\rangle\langle x\otimes\phi|B\big)\\
&=\operatorname{tr}\big((|y\rangle\langle x|\otimes|\psi\rangle\langle \phi|)B\big)\\
&=\operatorname{tr}\big(|y\rangle\langle x|\operatorname{tr}_{|\psi\rangle\langle \phi|}(B)\big)=\langle x,\operatorname{tr}_{|\psi\rangle\langle \phi|}(B) y\rangle\,.
\end{align*}
The second equality follows, e.g., from duality (see also Lemma~\ref{lemma_rewrite_trace_expression}).

(ii): It suffices to verify this equality on pure tensors as the span of those is dense in $\mathcal H\otimes\mathcal K$. Indeed for all $x\in\mathcal H$, $y\in\mathcal K$ one finds
\begin{align*}
\Big(\sum_{j\in J}\iota_{g_j}\iota_{g_j}^*\Big)(x\otimes y)&=\sum_{j\in J}\iota_{g_j}\langle g_j,y\rangle x\\
&=\sum_{j\in J}\langle g_j,y\rangle(x\otimes g_j)=x\otimes \Big(\sum_{j\in J}\langle g_j,y\rangle g_j\Big)=x\otimes y\,.
\end{align*}
In the last step we used the basis expansion formula \cite[Prop.~12.4]{MeiseVogt97en}.
\end{proof}

Now given any $A,B\in\mathbb C^{n\times n}$, applying \eqref{eq:Stinespring_second_der2} yields
\begin{align*}
\operatorname{tr}\big(B\ddot\Phi_0(A)\big)=&-\operatorname{tr}\big(B\operatorname{tr}_\omega(H^2)A\big)-\operatorname{tr}\big(BA\operatorname{tr}_\omega(H^2)\big)\\
&+2\operatorname{tr}\big(B\operatorname{tr}_{\mathcal K}\big( H  (A\otimes\omega) H \big)\big)\\
&=-\operatorname{tr}\big((AB\otimes\omega)H^2\big)-\operatorname{tr}\big((BA\otimes\omega)H^2\big)\\
&+2\operatorname{tr}\big((B\otimes\mathbbm1)H  (A\otimes\omega) H \big)\\
&=-\operatorname{tr}\big((A\otimes\omega)(B\otimes\mathbbm1)H^2\big)-\operatorname{tr}\big((B\otimes\mathbbm1)(A\otimes\omega)H^2\big)\\
&+2\operatorname{tr}\big((B\otimes\mathbbm1)H  (A\otimes\omega) H \big)\,.
\end{align*}
Inserting the expansions $\omega=\sum_{k\in N}r_k|g_k\rangle\langle g_k|$ as well as $\mathbbm1_{\mathcal K}=\sum_{j\in J}|g_j\rangle\langle g_j|$ \cite[Prop.~12.4]{MeiseVogt97en}, and making use of Lemma~\ref{lemma_partial_trace} we find that the previous expression is equal to
\begin{align*}
\sum_{(j,k)\in J\times N}r_k\Big(&-\operatorname{tr}\big((A\otimes  |g_k\rangle\langle g_k|)(B\otimes |g_j\rangle\langle g_j|)H^2\big)\\
&-\operatorname{tr}\big((B\otimes |g_j\rangle\langle g_j|)(A\otimes  |g_k\rangle\langle g_k|)H^2\big)\\
&+2\operatorname{tr}\big((B\otimes |g_j\rangle\langle g_j|)H  (A\otimes  |g_k\rangle\langle g_k|) H \big)\Big)\\
=
\sum_{(j,k)\in J\times N}r_k\Big(&-\operatorname{tr}\big(\iota_{g_k}A\iota_{g_k}^*\iota_{g_j}B\iota_{g_j}^*H^2\big)-\operatorname{tr}\big(\iota_{g_j}B\iota_{g_j}^*\iota_{g_k}A\iota_{g_k}^*H^2\big)\\
&+2\operatorname{tr}\big(\iota_{g_j}B\iota_{g_j}^*H\iota_{g_k}A\iota_{g_k}^* H \big)\Big)\,.
\end{align*}
One readily verifies $\iota_x^*\iota_y=\langle x,y\rangle\mathbbm1_n$ meaning the first two terms simplify to
\begin{align*}
-\sum_{k\in N}r_k\operatorname{tr}\big(\iota_{g_k}AB\iota_{g_k}^*&H^2\big)-\sum_{k\in N}r_k\operatorname{tr}\big(\iota_{g_k}BA\iota_{g_k}^*H^2\big)\\
=&
-\sum_{k\in N}r_k\operatorname{tr}\big(AB\iota_{g_k}^*H^2\iota_{g_k}\big)-\sum_{k\in N}r_k\operatorname{tr}\big(BA\iota_{g_k}^*H^2\iota_{g_k}\big)\\
=&
-\sum_{(j,k)\in J\times N}r_k\operatorname{tr}\big(AB\iota_{g_k}^*H\iota_{g_j}\iota_{g_j}^*H\iota_{g_k}\big)\\
&-\sum_{(j,k)\in J\times N}r_k\operatorname{tr}\big(BA\iota_{g_k}^*H\iota_{g_j}\iota_{g_j}^*H\iota_{g_k}\big)\,.
\end{align*}
Here we used a more general cyclicity property of the trace \cite[Lemma 3.1]{DvE18} as well as Lemma~\ref{lemma_partial_trace} (ii).
At this point we are almost done; all that is left is showing that $\operatorname{tr}(B{\bf\Gamma}_{V_{jk}}(A))$ (with $V_{jk}=\sqrt{2r_k}\operatorname{tr}_{|g_k\rangle\langle g_j|}(H)$ as above) for all $j\in J$, $k\in N$ equals
\begin{equation}\label{eq:app_D_left_to_show}
\begin{split}
r_k\operatorname{tr}(AB\iota_{g_k}^*H\iota_{g_j}\iota_{g_j}^*H\iota_{g_k}) + r_k\operatorname{tr}&(BA\iota_{g_k}^*H\iota_{g_j}\iota_{g_j}^*H\iota_{g_k}) \\
&- 2r_k\operatorname{tr} (B\iota_{g_j}^*H\iota_{g_k}A\iota_{g_k}^* H \iota_{g_j}  )\,.
\end{split}
\end{equation}
%As defined above $V_{jk}=\sqrt{2r_k}\operatorname{tr}_{|g_k\rangle\langle g_j|}(H)$.
Using Lemma~\ref{lemma_rewrite_trace_expression} (i) and Lemma~\ref{lemma_partial_trace}, expression \eqref{eq:app_D_left_to_show} comes out to be
\begin{align*}
r_k\operatorname{tr}\big(AB\operatorname{tr}_{|g_j\rangle\langle g_k|}&(H)\operatorname{tr}_{|g_k\rangle\langle g_j|}(H)\big) \\
&\qquad+ r_k\operatorname{tr}\big(BA\operatorname{tr}_{|g_j\rangle\langle g_k|}(H)\operatorname{tr}_{|g_k\rangle\langle g_j|}(H)\big)\\
&\qquad- 2r_k\operatorname{tr} \big(B\operatorname{tr}_{|g_k\rangle\langle g_j|}(H)A\operatorname{tr}_{|g_j\rangle\langle g_k|}(H) \big)\\
=\;&r_k\operatorname{tr}\big(AB\big(\operatorname{tr}_{|g_k\rangle\langle g_j|}(H)\big)^*\operatorname{tr}_{|g_k\rangle\langle g_j|}(H)\big) \\
&\qquad+ r_k\operatorname{tr}\big(BA\big(\operatorname{tr}_{|g_k\rangle\langle g_j|}(H)\big)^*\operatorname{tr}_{|g_k\rangle\langle g_j|}(H)\big)\\
&\qquad- 2r_k\operatorname{tr} \big(B\operatorname{tr}_{|g_k\rangle\langle g_j|}(H)A\big(\operatorname{tr}_{|g_k\rangle\langle g_j|}(H)\big)^* \big)\\
=\;& \frac12 \operatorname{tr}(ABV_{jk}^* V_{jk} )+\frac12\operatorname{tr}(BA V_{jk}^* V_{jk})-\operatorname{tr}(BV_{jk}A V_{jk}^*)\\
=\;& \frac12 \operatorname{tr}(BV_{jk}^* V_{jk} A)+\frac12\operatorname{tr}(BA V_{jk}^* V_{jk})-\operatorname{tr}(BV_{jk}A V_{jk}^*)
\end{align*}
which by \eqref{eq:lindblad_V} is $\operatorname{tr}(B{\bf\Gamma}_{V_{jk}}(A))$ as desired.
Finally, defining $\Psi\in\mathcal L(\mathbb C^{n\times n})$ via $\Psi(A):=2 \operatorname{tr}_{\mathcal K}\big( H  (A\otimes\omega) H \big)$
by~\eqref{eq:Stinespring_second_der2} yields $\ddot\Phi_0=\Psi-\frac12\Psi^*(\mathbbm1)(\cdot)-(\cdot)\frac12\Psi^*(\mathbbm1)$. 
Because $\Psi$ is completely positive (as composition of completely positive maps) it admits Kraus operators $\{V_j\}_{j=1}^\ell$ such that $\ell\leq n^2$ \cite[Rem.~6]{Choi75}. Thus $\ddot\Phi_0=-\sum_{j=1}^\ell{\bf\Gamma}_{V_j}$ which shows that the number of dissipative terms of $\ddot\Phi_0$ can be upper bounded by $n^2$ --- as well as $|J\times N|=\operatorname{rk}(\omega)\operatorname{dim}(\mathcal K)$.

(ii): Given a finite set of matrices $\{V_j\}_{j\in J}$ (w.l.o.g.~$\{V_1,\ldots,V_{|J|}\}$) $\subset\mathbb C^{n\times n}$ define $m:=|J|+1$ as well as
$$
H:=\frac{1}{\sqrt2}\sum_{j=1}^{m-1}\big(V_j\otimes|e_{j+1}\rangle\langle e_1| + V_j^*\otimes|e_{1}\rangle\langle e_{j+1}|\big)\in\mathbb C^{n\times n}\otimes\mathbb C^{m\times m}\,.
$$
Because $H$ is a Hermitian matrix, 
$\Phi(t):= \operatorname{tr}_{\mathbb C^m}(e^{iHt}((\cdot) \otimes |e_1\rangle\langle e_1|)e^{-iHt})$ is a type~I Stinespring curve. 
What we have to show now is that the this curve's second derivative at zero equals $-\sum_{j\in J}{\bf\Gamma_{V_j}}$.

Writing $\nu_j:=V_j\otimes|e_{j+1}\rangle\langle e_1| + V_j^*\otimes|e_{1}\rangle\langle e_{j+1}|$ for all $j=1,\ldots,m-1$, using \eqref{eq:Stinespring_second_der2} one finds
\begin{align*}\
\ddot\Phi(0)=\sum_{j,k=1}^{m-1}\Big( -\frac12\operatorname{tr}_{|e_1\rangle\langle e_1|}(\nu_j\nu_k)(\cdot)&-\frac12(\cdot)\operatorname{tr}_{|e_1\rangle\langle e_1|}(\nu_j\nu_k)\\
&+\operatorname{tr}_{\mathbb C^m}\big(\nu_j((\cdot)\otimes|e_1\rangle\langle e_1|)\nu_k\big) \Big)\,.
\end{align*}
The first two expressions simplify due to
\begin{align*}
\operatorname{tr}_{|e_1\rangle\langle e_1|}(\nu_j\nu_k)&=\operatorname{tr}_{|e_1\rangle\langle e_1|}(V_jV_k^*\otimes|e_{j+1}\rangle\langle e_{k+1}|+V_j^*V_k\otimes|e_1\rangle\langle e_1|\delta_{jk})\\
&=V_jV_k^*\operatorname{tr}\big(|e_1\rangle\langle e_1,e_{j+1}\rangle\langle e_{k+1}|  \big)  +\delta_{jk} V_j^*V_k\operatorname{tr}\big(|e_1\rangle\langle e_1,e_1\rangle\langle e_1|\big)\\
&=\delta_{jk} V_j^*V_k
\end{align*}
for all $j,k=1,\ldots,m-1$.
Similarly one finds
$$
\operatorname{tr}_{\mathbb C^m}\big(\nu_j((\cdot)\otimes|e_1\rangle\langle e_1|)\nu_k\big)=\operatorname{tr}_{\mathbb C^m}\big(V_j(\cdot)V_k^*\otimes|e_{j+1}\rangle\langle e_{k+1}|\big)=\delta_{jk}V_j(\cdot)V_k^*\,.
$$
Combining all of this yields
\begin{align*}
\ddot\Phi(0)&=\sum_{j,k=1}^{m-1}\Big( -\frac12\delta_{jk} V_j^*V_k(\cdot)-\frac12(\cdot)\delta_{jk} V_j^*V_k+\delta_{jk}V_j(\cdot)V_k^*\Big)\\
&=-\sum_{j=1}^{|J|}\Big( \frac12 \big(V_j^* V_j (\cdot)+(\cdot) V_j^* V_j\big)-V_j(\cdot) V_j^* \Big)=-\sum_{j\in J}{\bf\Gamma}_{V_j}
\end{align*}
which concludes the proof.

\bibliographystyle{mystyle}
\bibliography{../../../../control21vJan20.bib}

\end{document}